\documentclass{amsart}
\usepackage[utf8]{inputenc}
\usepackage[english]{babel}

\usepackage{amsmath}
\usepackage{amsfonts}
\usepackage{amssymb}
\usepackage{comment}
\usepackage{helvet}
\usepackage{courier}
\usepackage{graphicx}

\DeclareMathOperator{\Csp}{CSP}
\DeclareMathOperator{\Qcsp}{QCSP}
\DeclareMathOperator{\ic}{ic}
\DeclareMathOperator{\ci}{ci}
\DeclareMathOperator{\su}{su}

\DeclareMathOperator{\peak}{peak}
\DeclareMathOperator{\cyc}{cyc}

\DeclareMathOperator{\pp}{pp}
\DeclareMathOperator{\lele}{ll}
\DeclareMathOperator{\dlele}{dual-ll}

\DeclareMathOperator{\dpp}{dual-pp}

\DeclareMathOperator{\wave}{wave}

\newcommand{\Var}{Var}

\newcommand{\Aut}{Aut}
\newcommand{\Pol}{Pol}

\newcommand{\Q}{\mathbb{Q}}
\newcommand{\N}{\mathbb{N}}

\newcommand{\BetwC}{\text{BetwC}}
\newcommand{\CyclC}{\text{CyclC}}
\newcommand{\Betw}{\text{Betw}}
\newcommand{\Cycl}{\text{Cycl}}
\newcommand{\Sep}{\text{Sep}}
\newcommand{\EqXor}{\text{EqXor}}
\newcommand{\EqOr}[1]{\text{EqOr}_{#1}}

\newcommand{\forallexists}{\forall\!\exists\!\wedge}

\newcounter{thm}
\newtheorem{theorem}[thm]{Theorem}
\newtheorem{lemma}[thm]{Lemma}
\newtheorem{observation}[thm]{Observation}
\newtheorem{proposition}[thm]{Proposition}
\newtheorem{definition}[thm]{Definition}

\newcounter{tablecell}

\begin{document}

\title[Tractability Frontier for Dually-Closed  Temporal QCSPs]{Tractability Frontier for Dually-Closed  Temporal Quantified Constraint Satisfaction Problems}

\author{Micha{\l} Wrona}
	\address{Theoretical Computer Science Department, Jagiellonian University, Poland}
	\email{michal.wrona@uj.edu.pl}
	\urladdr{https://www.tcs.uj.edu.pl/wrona}
	\thanks{Micha{\l} Wrona is partially supported by National Science Centre, Poland grant number 2020/37/B/ST6/01179}
%\author{Micha{\l} Wrona}
%\institute{LIX (CNRS UMR 7161),
% \'{E}cole Polytechnique, 1128 Palaiseau, France\\
%Email: mwrona@lix.polytechnique.fr}

\maketitle              % typeset the title of the contribution

\begin{abstract}
A temporal (constraint) language is a relational structure with a first-order definition in 
the  rational numbers with the order. 
We study here the complexity of 
the Quantified Constraint Satisfaction Problem (QCSP)
for temporal constraint languages.

Our main contribution is a dichotomy for the restricted class of
dually-closed temporal languages. We prove that QCSP
for such a language is either solvable in polynomial time or 
it is hard for NP or coNP. Our result generalizes a similar dichotomy  
of QCSPs for equality languages~\cite{qecsps}, which are relational structures definable by 
Boolean combinations of equalities.
\end{abstract}

\section{Introduction}

The main goal of computational complexity is to characterize the difficulty of 
decision problems by placing them into complexity classes according to the time resources
or space resources required to solve the problem.
Much more convenient and mathematically elegant, however, is to work with formalisms 
such as Constraint Satisfaction Problems (CSPs)
or Quantified Constraint Satisfaction Problems (QCSPs), 
considered in this paper, 
that can express 
infinitely many decision problems.

The formalism of Constraint Satisfaction Problems express 
Boolean satisfiability problems: $3$-SAT, $2$-SAT; $k$-coloring
and solving equations over finite fields as well as
many problems in different branches of Artificial Intelligence such as 
Temporal Reasoning.
An instance of 
the Constraint Satisfaction Problem $\Csp(\Gamma)$ parametrized by a  relational structure
$\Gamma$ is a primitive-positive (pp) sentence
of the form:
$(\exists v_1 \ldots \exists v_n (R(v_{i_1}, \ldots, v_{i_k}) \wedge \ldots))$,
where the inner quantifier-free part is the conjunction of relational symbols,
from the signature of $\Gamma$,
with variables.

This framework on one hand allows to define many important decision problems. 
On the other hand,
 complexity classifications of the flavour of
Schaefer's theorem~\cite{Schaefer} these problems tend to display 
are of theoretical interest. Directly relevant to this paper are temporal CSPs
 that are $\Csp(\Gamma)$ for $\Gamma$ with a first-order definition in $(\Q; <)$, called \emph{temporal (constraint) languages}.
In~\cite{tcsps-journal}, nine large classes of tractable (solvable in polynomial time) temporal
CSPs has been identified
and it was proved that all other such problems are NP-complete.
These nine  tractable classes  contain many decision problems 
studied previously and scattered across the literature, e.g.,
the Betweenness and the Cyclic Ordering Problem mentioned in~\cite{GareyJohnson},
and network satisfaction problem for Point Algebra~\cite{PointAlgebra}. Furthermore, temporal $\Csp(\Gamma)$
has been studied in the literature for so-called Ord-Horn languages~\cite{Nebel}, and AND/OR precedence
constraints in scheduling~\cite{and-or-scheduling}. 

A natural generalization of the CSP is the Quantified Constraint Satisfaction Problem (QCSP) for a relational structure $\Gamma$,
denoted by $\Qcsp(\Gamma)$ that next to existential allows also universal quantifiers in the input sentence.
Similarly as the CSP, this problem has been studied widely  in the literature, see 
e.g.~\cite{BornerBJK03,Meditations}. In this paper, we study $\Qcsp(\Gamma)$ for temporal languages. 
Although a number of partial results
has been obtained~\cite{qecsps,CompleteEqualityQCSPs,ChenM12,CharatonikWronaLPAR,CharatonikWronaCSL,ChenW12,ChenBodirskyWrona,MFCSWrona14},
these efforts did not lead to a full complexity classification of all temporal $\Qcsp(\Gamma)$.
One of the reasons is that QCSPs are usually harder to classify than CSPs. This also holds in our case.
For instance, temporal CSPs are at most NP-complete, whereas temporal QCSPs 
can be at most PSPACE-complete. In more detail, 
nine tractable classes of temporal CSPs identified in~\cite{tcsps-journal}
are given by so-called polymorphisms that are in the case of temporal languages $\Gamma$
operations  from $\Q^k$ for some $k \in \N$ to $\Q$ that preserve $\Gamma$ 
(homomorphisms from $\Gamma^k$ to $\Gamma$).
The first of these classes is the class preserved by constant operations, the other 
polymorphisms (all binary) that come into play are named: $\min, \max, \text{mx}, \text{dual-mx}, \text{mi}, \text{dual-mi},
\lele,\dlele$. Although constant polymorphisms make CSP trivial, QCSP for a temporal 
language preserved by a constant polymorphism may be even PSPACE-complete~\cite{CharatonikWronaCSL}.
When it comes to $\min, \max, \text{mx}, \text{dual-mx}$, these operations provide tractability for both temporal
CSPs and QCSPs~\cite{ChenBodirskyWrona}, the complexity of temporal QCSPs preserved by $\text{mi}$ and $\text{dual-mi}$
is not known. But it is known that $\lele$ and $\dlele$ do not in general 
provide tractability for temporal QCSPs~\cite{qecsps,MFCSWrona14}. 
However, Guarded Ord-Horn Languages identified in~\cite{ChenW12}, 
are all preserved by both $\lele$ and $\dlele$ and provide
tractability for temporal QCSPs.
Other partial results obtained in the literature present somehow restricted classifications of temporal QCSPs.
Such classifications has been provided in the following three cases.
\begin{enumerate}
\item For a while it has been known~\cite{qecsps} (see~\cite{ChenM12} for a different proof) that a QCSP of an \emph{equality language} 
is in P, it is NP-complete or coNP-hard.
Recently the problem has been solved completely by showing that such a problem is always either in LOGSPACE, it is NP-complete
or PSPACE-complete~\cite{CompleteEqualityQCSPs}.
\item \emph{Positive languages} 
that are relational structures with a positive definition (first-order definition using $\wedge$, $\vee$,
and $\leq$ only, negation is not used) in $(\Q; \leq)$ 
were classified in~\cite{CharatonikWronaCSL,CharatonikWronaLPAR} where it was proved that
a corresponding QCSP is in LOGSPACE, is NLOGSPACE-complete, P-complete, NP-complete or PSPACE-complete.
\item 
A temporal language $\Gamma$ is dually-closed if in addition to a relation $R$ in $\Gamma$,
it also contains a relation $- R$ obtained from $R$ by replacing each tuple $t = (t[1], \ldots, t[n])$
with $-t = (-t[1], \ldots, -t[n])$ where $-$ is a unary operation that sends a rational number $q$ to $-q$.
A temporal language is Ord-Horn if it can be defined by a conjunction of so-called Ord-Horn clauses of the form:
$(x_1 \neq y_1 \vee \cdots \vee x_k \neq y_k \vee x R y)$ where $R \in \{<, \leq, = \}$ and
both the disjunction of disequalities and a literal $x R y$ can be omitted.
It was proved in~\cite{MFCSWrona14} that QCSP for a dually-closed Ord-Horn language is in P if it is Guarded Ord-Horn 
and it is coNP-hard otherwise.
\end{enumerate}
 
The main contribution of this paper is the tractability frontier of QCSPs for all dually-closed temporal languages. 
We prove that for every dually-closed temporal language $\Gamma$ the problem $\Qcsp(\Gamma)$ is either in $P$
or it is NP-hard or coNP-hard.
The result is built on all three classifications listed above and  clearly generalizes the third of the listed classifications. It also generalizes the orginal result on equality languages~\cite{qecsps}: see Section~\ref{sect:partclass} for details.
Although our proof is based on previously developed classifications, there are many
interesting languages that are not captured by any of them. For instance, 
the language $(\Q; \{(x,y,z) \mid (x = y = z) \vee (x < y < z) \vee (x > y > z) \}, <)$, 
where $<$ is a shortcut for $\{ (x, y) \in \Q^2 \mid x < y \}$,
is dually-closed but neither it is an equality language nor positive nor Ord-Horn.

To provide our classification we carry out a careful analysis of unary operations preserving
dually closed temporal languages. In particualar we analyse operations that are
non-injective. In the case of temporal CSPS, such analysis is not necessary
since all corresponding languages give rise to trivial and therefore 
tractable problems. But it is not true not only for temporal QCSPs but also for temporal abduction, whose complexity
has been investigated in~\cite{SchmidtW13,AAAIWrona14}. Similar anylysis will be also helpful in order to complete
the project of identifying temporal languages with the so-called local-to-global consistency, 
provided so far only for Ord-Horn languages~\cite{CPWrona12}.

\subsection{Outline.} 
We start with preliminaries in Section~\ref{sect:preliminaries}.
Partial results from the literature that we use in this paper are presented in
Section~\ref{sect:partclass}. In Section~\ref{sect:constantoperation}, we present first 
some dually-closed
temporal languages that give rise to hard QCSPs and then we use these results 
to analyse dually-closed temporal languages that are preserved by a constant operation.
We show for  corresponding QCSPs that either one of the hard problems reduces to it or
it is captured by one of the three classifications listed above.
In Section~\ref{sect:classs}, we complete
the classification of QCSPs for dually closed temporal constraint languages.

\begin{comment}
For instance $\Csp(\Gamma)$
for every Ord-Horn language $\Gamma$ may be easily solved in polynomial time by an appropriate 
Datalog program. On the other hand $\Qcsp(\Gamma)$ where $\Gamma$ is an Ord-Horn 
language $(\Q; \{ (x,y,z) \mid x = y \rightarrow y = z \})$ is known to be coNP-hard~\cite{qecsps}.

In this paper we study $\Qcsp(\Gamma)$ for a natural subclass of Ord-Horn constraint languages
that we call \emph{dually-closed Ord-Horn} languages. 
In particular every equality language~\cite{qecsps}, that is definable by a Boolean combination of $=$, is
dually-closed. 
The main result of the paper presented in Section~\ref{sect:mainresult} is a dichotomy:
we show that $\Qcsp(\Gamma)$ for such languages $\Gamma$ is either coNP-hard or it is in $P$. 
In the second case $\Gamma$ is a Guarded Ord-Horn language and the corresponding QCSP can be solved by an algorithm 
based on establishing local consistency presented in~\cite{ChenW12}.  
\end{comment}

\section{Preliminaries}
\label{sect:preliminaries}

We write $[n]$ for the set $\{1, \ldots, n \}$ with $n \in \N$ and $t = (t[1], \ldots, t[n])$
for an $n$-ary tuple $t$. The $i$-th entry of the tuple $t$, we denote by $t[i]$.

\subsection{Formulas and definability.}
We consider two restricted forms of first-order (fo)-formulas.
Let $\tau$ be a signature.
A first-order $\tau$-formula is a \emph{$\forallexists$-formula} if it has the form
$Q_1 v_1 \ldots Q_n v_n (\psi_1 \wedge \cdots \wedge \psi_m)$, where each $Q_i$ is a quantifier from $\{ \forall, \exists \}$,
and each $\psi_i$ is an atomic $\tau$-formula of the form $R(x_1, \ldots, x_k)$ where $R \in \tau$.
A \emph{primitive positive (pp)-formula} is a $\forallexists$-formula where all quantifiers are existential. 
A \emph{$\forallexists$-sentence (a pp-sentence)} is a $\forallexists$-formula (pp-formula) without free variables.

We say that a relation $R$ is (fo)-definable ($\forallexists$-, or pp-definable) in a relational structure
$\Gamma$ if $R$ has the same domain
as $\Gamma$, and there is a fo-formula ($\forallexists$-, or pp-formula) $\phi$ in the signature of $\Gamma$
such that $\phi$ holds exactly on those tuples that are contained in $R$.
A relational structure $\Delta$ is fo-definable in $\Gamma$
if every relation in $\Delta$ is fo-definable in $\Gamma$.

\subsection{Temporal languages and formulas.} 
In this paper a \emph{temporal formula}   
is a fo-formula built from quantifiers, logical connectivities
and relational symbols: $<,\leq, \neq, =$.
A \emph{temporal relation} is a relation with a fo-definition in 
$(\Q; <,\leq, \neq, =)$ and a \emph{temporal (constraint) language} is a relational structure over
a finite signature consisting of  the domain $\Q$ and a finite number of temporal relations.

\subsection{QCSPs.}
Let $\Gamma$ be a relational structure over a finite signature.
The \emph{Quantified Constraint Satisfaction Problem} for $\Gamma$, denoted by $\Qcsp(\Gamma)$, 
is the problem to decide if a given $\forallexists$-sentence over a signature of $\Gamma$ is 
true in $\Gamma$.

\begin{lemma}
(\cite{qecsps,BornerBJK03})
\label{lem:fewdef}
Let $\Gamma_1, \Gamma_2$ be constraint languages.  If $\Gamma_1$
has a $\forallexists$-definition in $\Gamma_2$, then $\Qcsp(\Gamma_1)$ is logarithmic space reducible to
$\Qcsp(\Gamma_2)$. 
\end{lemma}

The above lemma states that we can reduce between QCSPs problems by providing appropriate $\forallexists$-definitions.
In this paper  we restrict ourselves to pp-definitions but  we rather work with their characterization by polymorphisms
than with them directly.

\subsection{Polymorphisms.}
%The set of relations with a pp-definition in $\Gamma$ will be denoted by $\coclone{\Gamma}$.

Let $t_1, \ldots, t_k$ be a $m$-tuples over $\Q$ 
and let $f$ be a function (called also an operation) $f:D^k \rightarrow D$
Then we write $f(t_1 ,\ldots, t_k )$ for
the tuple obtained from the tuples $t_1 , \ldots, t_k$ 
by applying $f$ componentwise, i.e., for the $m$-tuple
$t =(f(t_1[1], \ldots , t_k[1]), \ldots , f (t_1[m],\ldots, t_k[m]))$.
If for all $t_1, \ldots, t_k \in R$ the tuple $t$ is also in $R$,
then we say that $f$ is a \emph{polymorphism} of (preserves) $R$. If it happens
for all relations in $\Gamma$, then $f$ is a polymorphism of (preserves) $\Gamma$.
If $f$ does not preserve $R$ or $\Gamma$, then we say that $f$ violates $R$ or $\Gamma$, respectively.

Observe that an \emph{automorphism}  is a  unary polymorphism that preserves all
relations  and their complements. The set of automorphisms of $\Gamma$
is denoted by $\Aut(\Gamma)$. An \emph{oribt} of a $k$-tuple $t$ 
wrt. $\Aut(\Gamma)$ is the set $\{ \alpha(t)  \mid \alpha \in \Aut(\Gamma)  \}$.

The set of all polymorphisms of a temporal language $\Gamma$, denoted by $\Pol(\Gamma)$,
forms an algebraic object called a \emph{clone}~\cite{Szendrei}, which is a set of operations defined on a set 
$\Q$ that is closed under composition and
that contains all projections. Moreover, $\Pol(\Gamma)$ is also closed under 
\emph{interpolation} (see Proposition
1.6 in~\cite{Szendrei}): we say that a $k$-ary operation f is \emph{interpolated} 
by a set of $k$-ary operations $F$ if for
every finite subset $A$ of $Q$ there is some operation $g \in F$ such that $f(a) = g(a)$ 
for every $a \in A^k$ .
We say that $F$ \emph{locally generates} an operation $g$ if $g$ 
is in the smallest clone that is closed under
interpolation and contains all operations in $F$. 

Since all temporal languages are preserved by all automorphisms of $(\Q;<)$,
the following definition makes sense.
We say that a set of operations $F$ \emph{generates} an
operation $g$ if $F$ together with all automorphisms of $(\Q; <)$ locally generates g. In case that $F$
contains just one operation $f$, we also say that $f$ generates $g$. 
We have the following easy observation.

\begin{observation}
\label{obs:ifgenthenclosed}
Let $\Gamma$ be a temporal language. If $\Gamma$ is preserved by all operations in $F$ and $F$ generates $g$, then
$g$ also preserves $\Gamma$. 
\end{observation}

To prove that operations in $F$ generate $g$, we use the following result from~\cite{tcsps-journal}.

\begin{lemma}
\label{lem:gen}
An operation $f$ generates $g$ if and only if every temporal relation that is preserved
by $f$ is also preserved by $g$.
\end{lemma}

For temporal languages we have the following chracterization of pp-definability
in terms of polymorphisms~\cite{BodirskyNesetrilJLC}.

\begin{theorem}
\label{thm:GaloisTemp}
Let $\Gamma_1, \Gamma_2$ be temporal languages, then $\Gamma_2$
has a pp-definition in $\Gamma_1$ if and only if $\Pol(\Gamma_1) \subseteq \Pol(\Gamma_2)$.
\end{theorem}

A direct consequence of Theorem~\ref{thm:GaloisTemp} is that if a temporal language $\Gamma_1$
does not have a pp-definition in $\Gamma_2$, then there is an operation over $\Q$
which is a polymorphism of $\Gamma_2$ but is not a polymorphism of $\Gamma_1$.
The next lemma~\cite{tcsps-journal} allows us to bound the arity of such polymorphism.

\begin{lemma}
\label{lem:arityred}
Let $\Gamma$ be a relational structure and let $R$ be a $k$-ary relation that is a union of $l$ orbits
of $k$-tuples of $\Aut(\Gamma)$. If $R$ is violated by a polymorphism $g$ of $\Gamma$ of arity $m \geq l$, then $R$ is also
violated by an $l$-ary polymorphism of $\Gamma$.
\end{lemma}

Provided we have some more knowledge about $\Pol(\Gamma)$, the 
following observation allows us to reduce the arity of a considered polymorphism 
even more. 

\begin{observation}
\label{obs:futherred}
Let  $t_1, \ldots, t_n, t \in \Q^m$. If $\Gamma$ is preserved by $f: \Q^n \rightarrow \Q$ and $g: \Q^{n-1} \rightarrow \Q$
such that $f(t_1, \ldots, t_n) = t$ and $g(t_1, \ldots, t_{n-1}) = t_n$ then $\Gamma$ is preserved 
by $h(x_1, \ldots, x_{n-1}) := f(x_1,\ldots, x_{n-1}, g(x_1, \ldots, x_{n-1}))$
such that $h(t_1, \ldots, t_{n-1}) = t$.
\end{observation}

\subsection{Operations of Special Importance for this Paper.}
Let $-: \Q \rightarrow \Q$ be the unary operation such that for every $q \in \Q$ we have $-(q) = - q$.
Let $f: \Q^n \rightarrow \Q$. The \emph{dual} of $f$ is the operation $-f(-x_1, \ldots, -x_n)$ 
denoted by $\overline{f}$. 

Let $e$ be any order-preserving bijection
between $(-\infty, \pi)$ and $(\pi, \infty)$. Then the operation $\cyc: \Q \rightarrow \Q$ is defined by $e(x)$ for $x < \pi$ 
and by $e^{-1}(x)$ for $x > \pi$.

We say that $f:\Q^k \rightarrow \Q$ for some $k \in \N$ is a \emph{constant operation} if
$f$ sends every $k$-tuple of rational numbers to the same rational number.

We define the operation $\wave: \Q \rightarrow \Q$ to be equal to $-x$ for $x < 0$,
to $0$ for $0 \leq x \leq 1$ and to $(x-1)$ for $x > 1$.
We define the operation $\peak$ to be equal to $-1$ for $x \neq 0$ and equal to $1$ for $x = 0$.
Further, we define $\su_i: \Q \rightarrow \Q$  where $i \in \N$ to be an operation that satisfies  $\su_i(x) = 0$ for $x < 0$;
$\su_{i}(x) = j$ for $[j-1, j)$ if $0 < j < i$; and $\su_i(x) = i$ for $[i-1, + \infty)$.
We define $\ic$ to be an operation such that $\ic(x) = x$ for all $(x < 0)$
and $\ic(x) = 0$ for $(x \geq 0)$; and $\ci$ to be an operation such that $\ci(x) = 0$ for $x < 0$
and $\ci(x) = x$ for $(x \geq 0)$. Observe that $\ic$ and $\ci$ are the duals of each other.

Let $\pp$ be an arbitrary binary operation on $\Q$ such that 
$\pp(a_1, b_1) \leq \pp(a_2, b_2)$ 
if and only if one of the following cases applies:
\begin{itemize}
\item $a_1 \leq 0$ and $a_1 \leq a_2$
\item $0 < a_1$, $0 < a_2$, and $b_1 \leq b_2$. 
\end{itemize}
All operation that satisfy the above conditions generate each other and therefore 
the same clone.
The dual of $\pp$ is known as $\dpp$.

Let $\lele$ be a binary operation on $Q$ such that $\lele(a_1, b_1) < \lele(a_2, b_2)$ if
\begin{itemize}
\item $a_1 \leq 0$ and $a_1 < a_2$, or
\item $a_1 \leq 0$ and $a_1 = a_2$ and $b_1 < b_2$ , or
\item $a_1, a_2 > 0$ and $b_1 < b_2$, or
\item $a_1 > 0$ and $b_1 = b_1$ and $a_1 < a_2$.
\end{itemize}
All operations satisfying these conditions  generate each other and the
same clone. The dual of $\lele$ is known as $\dlele$.

\section{Partial Classifications}
\label{sect:partclass}

In this section we present some results from the literature that we use for our classification.

\subsection{Equality Languages}
We say that a temporal language $\Gamma$ is an \emph{equality language} if it is definable
by Boolean combinations of equalities ($=$). A temporal language is an equality language if and only if
$\Pol(\Gamma)$ contains all permutations of $\Q$, i.e., all bijunctive functions from $\Q$ to $\Q$.

\subsection{Dually-Closed Temporal Languages.}
A temporal language $\Gamma$ is \emph{dually-closed} if  $f \in \Pol(\Gamma)$
always whenever $\overline{f} \in \Pol(\Gamma)$. 
Since $-$ is a permutation of $\Q$, we have that it preserves every equality language and hence every equality language is also dually-closed.

Let $R \subseteq \Q^n$. The dual of $R$ is the relation $\{(-t[1], \ldots, -t[n]) \mid t \in R \}$
denoted by $\overline{R}$. It is easy to prove~\cite{MFCSWrona14} that $\Gamma$ is a dually-closed language if
and only if the following holds:  $\Gamma$ pp-defines 
$R$ if and only if it pp-defines $\overline{R}$.

\subsection{Ord-Horn and Guarded Ord-Horn Languages}
Recall from the introduction that a 
temporal language is Ord-Horn if it can be defined by a conjunction of Ord-Horn clauses of the form:
$(x_1 \neq y_1 \vee \cdots \vee x_k \neq y_k \vee x R y)$ where $R \in \{<, \leq, = \}$ and
both the disjunction of disequalities and a literal $x R y$ can be omitted.
They have the following algebraic characterization~\cite{CPWrona12}.

\begin{proposition}
\label{prop:OHalg}
Let $\Gamma$ be a temporal language. Then $\Gamma$ is Ord-Horn if and only if it is preserved by both
$\lele$ and $\dlele$. 
\end{proposition}

The only tractable case for dually-closed $\Qcsp(\Gamma)$ is where $\Gamma$ is a guarded Ord-Horn
language. 

\begin{definition}
We say that a temporal language $\Gamma$ is Guarded Ord-Horn (GOH) if every relation in $\Gamma$
is definable by a GOH formula defined as follows.

\begin{enumerate}
\item A Basic OH formula which  is in one of the following forms is a GOH formula:
\begin{itemize}
\item $x = y$, $x \leq y$,
\item $(x_1 \neq y_1 \vee \cdots \vee x_p \neq y_p)$, or
\item $(x_1 \neq x_2 \vee \cdots \vee x_1 \neq x_q) \vee (x_1 < y_1) \vee (y_1 \neq y_2 \vee \cdots \vee y_1  \neq y_{q'})$.
\end{itemize}
\item If $\psi_1$ and $\psi_2$ are GOH formulas, then $\psi_1 \wedge \psi_2$
is a GOH formula.
\item If $\psi$ is a GOH formula, then
\begin{center}
$(x_1 \leq y_1) \wedge \ldots \wedge (x_m \leq y_m) \wedge$\\
$(x_1 \neq y_1 \vee \ldots \vee x_m \neq y_m \vee \psi)$
\end{center}
is a GOH formula. 
\end{enumerate}
\end{definition}

What might be interesting is that the algorithm for the following tractability result uses  
local consistency methods.

\begin{theorem}
\label{thm:GOHtract}
(\cite{ChenW12})
Let $\Gamma$ be a GOH structure. Then $\Qcsp(\Gamma)$ is solvable in polynomial time.
\end{theorem}

 \subsection{Dually-Closed Ord-Horn Languages}

For temporal languages that are both 
dually-closed and Ord-Horn we have the following dichotomy.

\begin{theorem}
\label{thm:duallyOH}
 Let $\Gamma$ be a dually-closed Ord-Horn temporal language. Then $\Qcsp(\Gamma)$
is in P if $\Gamma$ is Guarded Ord-Horn. Otherwise, it is coNP-hard.
\end{theorem}

\subsection{Temporal Languages}

We say that a temporal language is \emph{positive} if every relation in $\Gamma$
has a positive definition in $(\Q; \leq)$, a first-order definition with use of $\wedge, \vee$, and $\leq$ 
only. No negation is in use. The following algebraic characterization of positive temporal languages 
is easy to prove.

\begin{proposition}
\label{prop:positivewave}
 Let $\Gamma$ be a temporal language, then $\Gamma$ is positive if and only if it is preserved
by $\wave$.
\end{proposition}

The following dichotomy is a consequence of the classification in~\cite{CharatonikWronaCSL,CharatonikWronaLPAR}.

\begin{theorem}
 \label{thm:positivefrontier}
Let $\Gamma$ be a positive temporal language. 
Then  $\Qcsp(\Gamma)$ is in P if it is preserved either by  $\pp$ or $\dpp$. Otherwise it is
NP-hard.
\end{theorem}

\section{Languages Preserved by a Constant Operation}
\label{sect:constantoperation}

First we present some temporal languages that give rise to coNP-hard and NP-hard
QCSPs.

\begin{theorem}
\label{thm:somehard}
Consider the following relations:
\begin{itemize}
\item $\BetwC = \{ (x,y,z) \in \Q^3 \mid (x < y < z \vee x > y > z \vee x = y = z)\}$;
\item $\CyclC = \{ (x,y,z) \in \Q^3 \mid (x < y < z \vee  y < z < x \vee z < x < y \vee x = y = z)\}$;
\item $\EqXor = \{ (x,y,z) \in \Q^3 \mid (x = y \vee x = z)  \}$ ;
\item $\EqOr{n} = \{ (x_1, \ldots, x_n) \in \Q^n \mid (\bigvee_{i,j \in [n], i \neq j} x_i = x_j ) \}$;
\item $S = \{ (x,y,z) \in \Q^3 \mid (x = y = z) \vee (x \neq y \wedge x \neq z \wedge z \neq y)\}$.
\end{itemize}
It holds that $\Qcsp(\Q;I)$; $\Qcsp(\Q;S)$, $\Qcsp(\Q;\BetwC)$ and $\Qcsp(\Q;\CyclC)$ are coNP-hard;
and $\Qcsp(\Q;\EqXor)$ as well as $\Qcsp(\Q;\EqOr{n})$ for all $n \geq 3$ are NP-hard.
\end{theorem}

\begin{proof}
To prove the theorem, we need one axiliary relation: 
$I = \{ (x,y,z) \mid (x = y \rightarrow y = z)  \}$.
By Theorem~5.5 in~\cite{qecsps}, it follows that $\Qcsp(I)$ and $\Qcsp(S)$
are coNP-hard as well as that $\EqXor$ and $\EqOr{n}$ for all $n \geq 3$ are NP-hard.
To prove that $\Qcsp(\BetwC)$ and $\Qcsp(CyclC)$ are coNP-hard, we will show that they pp-define 
$I$. Indeed, let  $R \in \{ \BetwC, \CyclC \}$. We claim that 
$I(x,y,z) = \exists u \exists v~(R(x,y,u) \wedge R(x,y,v) \wedge R(u,v,z))$.
In both cases, we have that if $x$ and $y$ have the same value, then also $u$ and $v$
 and in consequence also $z$ has the same value.
We have now to show that if $x$ has a different value than $y$, then $z$
can have an arbitrary value. 

We first consider the case where $x$ has the value less than $y$. If $R$ is $\BetwC$, then $u$ and $v$
have to be greater than $y$ and $u$ can be greater than $v$, and hence $z$ can be greater than $y$, equal to $y$,
between $x$ and $y$, equal to $x$ as well as less than $x$. 
If $R$ is $\CyclC$, then we can have  $x$ to be less than $y$ less than $u$ less than $v$ but still
$z$ can be less than $x$, equal to $x$, between $x$ and $y$, equal to $y$ as well as greater than $y$.

The second case we consider is when $y$ has the value less than $x$. If $R$ is $\BetwC$ then the analysis is symmetrical
to that in the previous paragraph since this relation is preserved by $-$. If $R$ is $\CyclC$, then 
we can place $y,u$, and $v$ so that $y < u < v < x$. Now, it is straighforward to check that 
$z$ can be between $v$ and $x$, equal to $x$, greater than $x$, less than $y$, equal to $y$ but also between $y$
and $u$.
\end{proof}

We will now show that  if a dually-closed temporal constraint language is preserved by a constant operation, then
$\Qcsp(\Gamma)$ pp-defines 
one of relations in Theorem~\ref{thm:somehard} 
and $\Qcsp(\Gamma)$ is hard, or $\Gamma$ is an equality language, a positive language or a dually-closed Ord-Horn language.
We start with an auxiliary lemma. 

%The proof of this lemma as well as all omitted proof can be found in the appendix.

\begin{lemma}
\label{lem:noncgen}
Let $f:\Q \rightarrow \Q$ be an operation that is neither constant nor it preserves $<$, then one of the following holds.
\begin{enumerate}
\item \label{noncgen:inj} The operation $f$ is injective and then:
\begin{enumerate}
\item $f$ 
 preserves $\Betw = \{(x,y,z) \in \Q^3 \mid (x < y < z \vee x > y > z) \}$ and generates $-$; or
\item $f$ preserves $\Cycl = \{(x, y, z) \in \Q^3 \mid (x < y < z \vee  y < z < x \vee z < x < y) \}$
and generates $\cyc$; or 
\item $f$ preserves $\Sep = \{(x_1,y_1,x_2,y_2) \in \Q^4 \mid  (x_1 < x_2 < y_1 < y_2) \vee (x_1 < y_2 < y_1 < x_2) \vee (y_1 < x_2 < x_1 < y_2) 
\vee (y_1 < y_2 < x_1 < x_2) \vee (x_2 < x_1 < y_2 < y_1) \vee (x_2 < y_1 < y_2 < x_1) \vee (y_2 < x_1 < x_2 < y_1) \vee (y_2 < y_1 < x_2 < x_1) \}$ 
and generates both $-$ and $\cyc$; or
\item $f$ generates all permutations.
\end{enumerate}
\item \label{noncgen:inf} The operation $f$ is of infinite image but is not injective and then generates $\ci$ or $\ic$.
\item \label{noncgen:fin} The operation $f$ is of finite image and then it generates $\su_1$ or $\peak$.
\end{enumerate}
\end{lemma}

\begin{proof}
The proof consists of three parts which correspond to three items from the formulation of the lemma.

\textit{(Part One)}
If $f$ is injective, then by the arguments in the proofs
of Propositions~17 and~19 in~\cite{tcsps-journal}, it is generated by automorphisms of $(\Q; <)$.
Since $f$ does not preserve $<$, 
by Cameron's theorem~\cite{Cameron5} (see also Theorems~4 and~13 in~\cite{tcsps-journal}), 
we have that one of the items in Case~\ref{noncgen:inj} holds.

\textit{(Part Two)}
We now turn to the case where $f$ is of infinite image but is not injective. Before we continue, we make an observation.

\begin{observation}
\label{obs:incdecseq}
Let $f: \Q \rightarrow \Q$ be an operation  that takes infinitely many values in $(a, b)$ 
where $a, b \in \Q \cup \{ -\infty, + \infty \}$, then there is an infinite sequence of rational
numbers contained in $(a, b)$:
\begin{enumerate}
 \item \label{incdecseq:dec} either of the form $q_1 > q_2 > \cdots$ such that:
\begin{enumerate}
\item \label{incdecseq:decdec} $f(q_1) > f(q_2) > \cdots$,
\item \label{incdecseq:decinc} $f(q_1) < f(q_2) < \cdots$;
\end{enumerate}
\item \label{incdecseq:inc} or of the form $q_1 < q_2 < \cdots$ such that:
\begin{enumerate}
\item \label{incdecseq:incdec} $f(q_1) > f(q_2) > \cdots$,
\item \label{incdecseq:incinc} $f(q_1) < f(q_2) < \cdots$;
\end{enumerate}
\end{enumerate}
%In Cases~\ref{incdecseq:decinc} and~\ref{incdecseq:incdec}, the operation $f$ generates $-$.
\end{observation}

\begin{proof}
Let $S$ be the set of elements in $(a,b)$ such that for all $x,y \in (a,b)$ we have $f(x) \neq f(y)$
and let $c \in S$. Now, either $S_1 := S \cap (a,c)$ or $S_2 := S \cap (c,b)$. We consider only the second case
in which we have that either for every $q \in S_2$ there is $p > q$ and we are in Case~\ref{incdecseq:inc}; or
there is $p \in S_2$ such that for every $\delta > 0$, there is $q \in S_2$ such that $p < q$ 
and $\left| p -q \right| \leq \delta$ and we are in Case~\ref{incdecseq:dec}.
Now an easy application of Infinite Ramsey Theorem gives us subsequence which is either strictly increasing or
strictly decreasing.  

We now show that in Cases~\ref{incdecseq:decinc} and~\ref{incdecseq:incdec}, the operation $f$ generates $-$.
We convey the proof only in the second case. We use Lemma~\ref{lem:gen}. Let $t$ be an $n$-ary tuple 
(of some $n$-ary relation) such that 
$p_1< \ldots < p_k$ are all pairwise different values in $t$. It is enough to observe that
$-(t) = \beta(f(\alpha(t)))$ where $\alpha, \beta \in \Aut(Q;<)$ are such that $\alpha$
sends $p_1, \ldots, p_k$ to $q_1, \ldots, q_k$ and $\beta$ sends $f(q_1), \ldots, f(q_k)$
to $- (p_1), \ldots, -(p_k)$.
\end{proof}

Let now $f$ be a non injective function that takes infinitely many values for arguments in the interval $(a_1, a_2)$
with $a_1, a_2 \in \Q \cup \{ -\infty, \infty\}$.
Then we are in one of the four cases of Observation~\ref{obs:incdecseq}.
To simplify the proof, we will always assume that we are either in Case~\ref{incdecseq:incinc}
or in Case~\ref{incdecseq:incdec}. The proofs for two other cases are always similar.
 
\begin{comment}
If Case~\ref{incdecseq:decinc} or~\ref{incdecseq:incdec} holds, then $f$ generates $-$, and we are done. 
Thus, always when we consider such an interval, we can assume that we are in one of the two other cases.
To simplify the proof, in such a situation, we always assume that Case~\ref{incdecseq:incinc} holds.
The proof for the other case is always similar.
\end{comment}

Since $f$ is non-injective, there are $c_1 < c_2$ such that $f(c_1) = f(c_2)$.
Furthermore, $f$ takes infinitely many values either for arguments less than $a_1$, or arguments greater than $a_1$.
In the first case there are infinitely many values taken above $f(a_1)$ or below $f(a_1)$.
By the previous assumptions, we are either in Case~\ref{incdecseq:incinc}
or Case~\ref{incdecseq:incdec}. In the second case $f$ generates $-$.
In any case, either by considering $g := f$ or $g := -f$, we have
one of the two following situations. 
In the first situation, there is a sequence of rational numbers 
$q_1 < q_2 < \cdots < c_1 < c_2$ such that 
$g(q_1) < g(q_2) < \cdots < g(c_1) = g(c_2)$.
In the second situation there is a sequence of rational numbers
$q_1 < q_2 < \cdots < c_1 < c_2$ such that 
$g(c_1) = g(c_2) < g(q_1) < g(q_2) < \cdots$.
In both situations the operation $g$ generates $\ic$.
We provide the proof only for the first one. 
We use Lemma~\ref{lem:gen}.
Let $t$ be a tuple of length $n$ in some relation $R$ with pairwise different values 
$p_1 < \ldots < p_k < 0 \leq p_{k+1} < \cdots < p_{k + l}$
 We prove that $\ic(t)$ is in $R$ by 
the induction
on $l$. If $l = 0$, then $\ic(t) = \alpha(t)$ where $\alpha \in \Aut(\Q;<)$
is the identity on $\Q$. Assume now that we are done for $l = m$.
We will prove the claim for $l = (m+1)$. Let $t$ be a tuple in $R$
with pairwise different values $p_1 < \cdots < p_k < 0 \leq p_{k+1} < \cdots < p_{k+m+1}$
and $\alpha \in \Aut(Q;<)$ such that $\alpha(p_1) < \cdots < \alpha(p_k) 
< \alpha(p_{k+1}) < 0 \leq \alpha(p_{k+2}) < \cdots < \alpha(p_{k+m+1})$.
By the induction hypothesis, the relation $R$ contains a tuple $t_1 = \ic(\alpha(t))$
with values $r_1 < \cdots < r_k < r_{k+1} < r_{k+2} = 0$,
where for all $i \in [k+1]$ we have $r_i = \ic(\alpha(p_i))$
and $r_{k+2}  = \ic(\alpha(p_{k+2})) = \cdots = \ic(\alpha(p_{k+m+1})) = 0$.
Let now $\beta \in \Aut(\Q;<)$ be such that for $i \in [k]$
it holds $\beta(r_i) = q_i$ and $\beta(r_{k+1}) = c_1$ and $\beta(r_{k+2}) = c_2$, 
and $\gamma \in \Aut(\Q;<)$ be such that for all $i \in [k]$ 
we have $\gamma(g(q_i)) = p_i$ as well as $\gamma(g(c_{1})) = \gamma(g(c_{2})) = 0$. 
Observe that $\gamma(g(t_1))$
is equal to $\ic(t)$. It follows that $R$ contains $\ic(t)$ and we are done.

From now on, we assume that $f$ takes finitely many values for arguments in the interval $(- \infty, c_1)$
and infinitely in the interval $(c_1,\infty)$.  
Again, we consider only Cases~\ref{incdecseq:incinc} and~\ref{incdecseq:incdec}
from Observation~\ref{obs:incdecseq}. In the first situation we look at $g := f$
and in the second situation we look at $g := -f$.
In any case we have either a sequence of rational numbers 
$d_1 < d_2 < c_1 < q_1 < q_2 < \cdots$ such that
either  
$g(d_1) = g(d_2) < g(q_1) < g(q_2) < \cdots$ or
$g(q_1) < g(q_2) < \cdots < g(d_1) = g(d_2)$.
In both cases we have that $g$ generates $\ci$.
The proof is similar to the proof that $g$ generates $\ic$
from the previous paragraph.
This completes the second part of the proof of the lemma.

\textit{(Part Three)} If $f$ is of finite image, then  either there are two intervals $[a_1,a_2]$ 
and $[a_3,a_4]$ with $a_1 < a_2 < a_3 < a_4$ 
as well as $b_1 \neq b_2$
such that
for all $x \in [a_1, a_2]$ we have $f(x) = b_1$ and for all $x \in [a_3,a_4]$ it holds
$f(x) = f(y) = b_2$ or almost all rational numbers are sent by $f$ to the same value.
We will show that in the first case, the operation $f$ generates $\su_1$, whereas in the second case it generates $\peak$.
We start with first case. Observe that without  loss of generality we can assume that $b_1 < b_2$. Indeed, 
if $b_1 > b_2$, then instead of $f$ we consider  $g = f(\alpha(f))$ where $\alpha \in \Aut(\Q; <)$ satisfies $\alpha(b_1)  = a_3$
and $\alpha(b_2) = a_1$. The operation $g$ satisfies $g(x) = a_1$ for all $x \in [a_1, a_2]$ 
and $g(x) = a_3$ for all $x \in [a_3,a_4]$. 
By Lemma~\ref{lem:gen}, it is enough  to show that every relation preserved by $f$ is also preserved by $\su$.
Let $R$ be any relation preserved by $f$ and $t$ a tuple in $R$. We have to show that $t_1 = \su(t)$ is 
also a tuple in $R$. Let $q_1, \ldots, q_k < 0 \leq q_{k+1} < \cdots < q_l$ be pairwise different values 
occurring in $t$; $\alpha \in \Aut(\Q: <)$ such that $\alpha$ sends $q_1, \ldots, q_k$ to an interval 
$[a_1,a_2]$ and $q_{k+1}, \ldots, q_l$ to $[a_3, a_4]$ and  $\beta \in \Aut(\Q: <)$ 
such that $\beta(b_1) = 0$ and $\beta(b_2) = 1$. Observe that $\beta(\su(\alpha(t))) = t_1$.
It follows that $f$ generates $\su_1$.

We now consider the case where  $f$ sends almost all values to $b$. Since $f$ is not a constant operation,
there is $a_1 \in \Q$ such that $f(a_1) = b_1 \neq b$. Again, without loss of generality
we can assume that $b < b_1$. If it is not the case, then instead of $f$ we consider $f(\alpha(f))$
where $\alpha \in \Aut(Q;<)$ satisfies $\alpha(b_1) = a$ and $\alpha(b) = a_1$ for some $a < a_1$ with $f(a) = b$.
Let now $R$ be a relation preserved by $f$ and $t \in R$, by Lemma~\ref{lem:gen}, 
we have to show that $\peak(t)$ is in $R$. Let $q_1, \ldots, q_k, q_{k+1}, \ldots, q_l$ be
pairwise different values in $t$ different than $0$ and such that $q_1 < \cdots < q_{k} < 0 < q_{k+1} < \cdots < q_l$.
Let $\alpha \in \Aut(Q; \leq)$ be such that it sends $q_1, \ldots, q_k$ to rational numbers less than $a_1$
such that $f(\alpha(q_1)) = \cdots = f(\alpha(q_k)) = b$, it sends $0$ to $a_1$ and 
$q_{k+1}, \ldots, q_l$ to rational numbers greater than $a_1$
such that $f(\alpha(q_{k+1})) = \cdots = f(\alpha(q_l)) = b$. Further, let $\beta \in \Aut(Q;\leq)$
satisfy $\beta(b) = -1$ and $\beta(b_1) = 1$. Observe that $\beta(f(\alpha(t))) = \peak(t)$.
It follows that $f$ generates $\peak$.
\end{proof}

Now, we make a first serious step. We show that either $\Gamma$ pp-defines $\BetwC$, which by 
Theorem~\ref{thm:somehard} gives rise to the hard $\Qcsp$ or $\Gamma$ is preserved by one of few polymorphisms.

\begin{lemma}
\label{lem:violBetwCconst}
Let $\Gamma$ be a temporal language preserved by a constant operation. 
If $\Gamma$ does not pp-define $\BetwC$, then 
$\Gamma$  is preserved by $\lele, \dlele, \pp$ or $\dpp$,
$\su_1$,$ \peak$, $\ic$, $\ci$, $-$, $\cyc$ or all permutations.
\end{lemma}

\begin{proof}
If $\Gamma$ does not pp-define $\BetwC$, then by Theorem~\ref{thm:GaloisTemp}, it follows that
there is an operation $f$ that preserves $\Gamma$ and violates $\BetwC$. By Lemma~\ref{lem:arityred},
we can assume that $f$ is a ternary operation such that for tuples $t_1, t_2, t_3$ satisfying 
$t_1[1] < t_1[2] < t_1[3]$, $t_2[1] > t_2[2] > t_2[3]$ and $t_3[1] = t_3[2] = t_3[3]$
we have $f(t_1,t_2,t_3) = t$ and $t \notin \BetwC$.
Since $\Gamma$ is preserved by a constant operation, it is preserved by all constant operations, in particular
some constant operation $g$ satisfying $g(t_1, t_2) = t_3$. By Observation~\ref{obs:futherred},
we have that $\Gamma$ is preserved by an operation $h: \Q^2 \rightarrow \Q$ such that $h(t_1, t_2) = t$.

Observe that the operation $h$ violates $\Betw$ defined in the formulation
of Lemma~\ref{lem:noncgen}.
If $h$ preserves $<$, then 
by Lemma~49 in~\cite{tcsps-journal}, we have that $g$ generates $\pp$, $\dpp$, $\lele$, or $\dlele$.
From now on, we assume that $h$ violates $<$. We now consider the situation where
$h$ violates $\leq$. If it is the case, then there are tuples $s_1, s_2$
such that $s_1[1] < s_1[2]$, $s_2[1] = s_2[2]$, and such that for $s = h(s_1,s_2)$
we have  $s[1] > s[2]$. Since $\Gamma$ is preserved by a constant operation we can again 
apply Observation~\ref{obs:futherred} and obtain that there is a unary operation $h_s$
preserving $\Gamma$ such that $h_s(s_1) = s$. Since $h_s$ is neither a constant operation nor
it preserves $<$, it follows by Lemma~\ref{lem:noncgen} that $h_s$ generates 
$\su_1, \peak, \ic, \ci, -, \cyc$ or all permutations.
From now on we assume that $h$ preserves $\leq$. Let $h_{i} := h(x,x)$. If $h_i$ 
is neither a constant operation nor it preserves $<$, then we are again done 
by Lemma~\ref{lem:noncgen}. Thus, we have two cases to consider to complete the proof.

First we look at the situation where $h_i$ is a constant operation that sends all the rational 
numbers to $a$. Since $t \notin \BetwC$, there are $i \neq j$ in $[3]$ such that
$t[i] \neq t[j]$. It follows that either $t[i]$ or $t[j]$ is different than $a$. Assume 
without loss o generality that $t[i] = b \neq a$. Let $d_1, d_2 \in \Q$ 
be such that $d_1 < \min(t_i[1], t_i[2]) < d_2$ and $\alpha \in \Aut(\Q;<)$ be such that
it sends all rational numbers $q \in (- \infty, d_1) \cup (d_2, \infty)$ to 
$q$;
the interval $[d_1, t_i[2])$ to $[d_1, t_i[1])$,
and $[t_i[2], d_2]$ to $[t_i[1], d_2]$. Due to Cantor's theorem such $\alpha$ clearly exists.
Consider  $h_a(x) = h(x,\alpha(x))$. Observe that there is an infinite sequence 
 $q_1 < q_2 < \cdots < t_i[1] < p_1 < p_2 < \cdots$ of rational numbers 
 such that $h_a(q_1) = h_a(q_2) = \cdots = h_a(p_1) = h_a(p_2) = \cdots = a$ and $h(t_i) = b$. 
Now as in the second paragraph of the third part of the proof of Lemma~\ref{lem:noncgen},
we can show that $h_a$ and in consequence $f$ generates $\peak$.

The last case to consider is where $h_i$ preserves $<$. Since $h$ violates $<$ and 
preserves $\leq$, there are $c_1, c_2 \in \Q^2$ such that $c_1[i] < c_2[i]$
for $i \in [2]$ and $f(c_1) = f(c_2)$. Let $d_1, d_2 \in \Q^2$ 
be such that $d_1 < \min(c_1[1], c_1[2]) < \max(c_2[1], c_2[2]) < d_2$ and
$\alpha \in \Aut(\Q;<)$ be such that
it sends all rational numbers  $q \in (- \infty, d_1) \cup (d_2, \infty)$ to $q$;
the interval $[d_1, c_2[1])$ to $[d_1, c_1[1])$;
the interval $[c_2[1], c_2[2])$ to $[c_1[1], c_1[2])$;
the interval $[c_2[2],d_2]$ to $[c_2[1], d_2]$.
Such $\alpha$ clearly exists.
Consider  $h_a(x) = h(x,\alpha(x))$. 
Since $h_a$ is not injective and not constant, it follows by Lemma~\ref{lem:noncgen}
that $h_a$ and in consequence $f$ generates $\peak, \su_1, \ic, \ci, -, \cyc$, or all
permutations. 
\end{proof}

From now on, we provide lemmas which takes care of polymorphisms listed in Lemma~\ref{lem:violBetwCconst}.
We show that in each of these cases, we can reduce our classification to the existing ones. We first take a look
at the situation where $\Gamma$ is preserved by $\lele, \dlele, \pp$ or $\dpp$.

\begin{lemma}
\label{lem:llppOH}
Let $\Gamma$ be a dually-closed temporal language  such that $\Gamma$ 
is preserved by $\pp, \dpp, \lele,$ or $ \dlele$.
Then $\Gamma$ is a dually-closed Ord-Horn constraint language.
\end{lemma}

\begin{proof}
Since $\Gamma$ is dually-closed, we have two cases to consider.
If $\Gamma$ is preserve by $\lele$ or $\dlele$,
then it is preserved by the both operations. It follows
by Proposition~\ref{prop:OHalg}, then $\Gamma$ is an Ord-Horn 
language.

It remains to consider the case where $\Gamma$
is preserved by both $\pp$ and $\dpp$.
By~\cite{ChenBodirskyWrona}, if $\Pol(\Gamma)$ contains $\pp$ then every $R$ in $\Gamma$ can be defined as
a conjunction of clauses of the form:
\begin{equation}
\label{eq:ppclause}
x \neq y_1 \vee \cdots \vee x \neq y_k \vee x \geq z_1 \vee \cdots \vee x \geq z_l.
\end{equation}
The language $\Gamma$ is also preserved by $\dpp$. In that case as we show every clause~\ref{eq:ppclause} in the definition of every
relation in $\Gamma$ satisfies $l = 1$.
Suppose not. Then there is $R$ in $\Gamma$ that does not have an Ord-Horn definition.
Let $\phi$ be a definition of $R$ in terms of clauses of the form~(\ref{eq:ppclause})
with a minimal number of literals and let $\psi$ be clause in $\phi$ for which $l \geq 2$.
By the minimality of $\phi$, it is satisfied by assignments $t_1, t_2: \Var(\phi) \rightarrow \Q$
such that $t_i$ violates all literals in $\psi$ except for $x \geq z_i$.
Let $\alpha_1 \in \Aut(\Q;<)$ be such that 
$\alpha_1(t_1(z_1)) \leq \alpha_1(t_1(x)) \leq 0 < \alpha_1(t_1(z_2))$.
To reach the contradiction, we will show that $f := \dpp(\alpha_1(t_1), t_2)$
does not satisfy any disjunct of $\psi$. By the definition of $\dpp$,
it follows that $f(x) < f(z_1)$ and $f(x) < f(z_2)$. All other literals are violated since $\dpp$ preserves $<$ and $=$.
It completes the proof of the lemma.
\end{proof}

The operations $\ic$ and $\ci$ are the duals to each other and as we show in the proof 
of Theorem~\ref{thm:constthm}, $\{ \ic, \ci \}$  generates $\su_1$. Since we restrict ourselves to dually-closed languages,
the next lemma applies also to the case where 
$\Gamma$ is preserved by $\ic$ or $\ci$.

\begin{lemma}
\label{lem:su}
Let $\Gamma$ be an operation preserved by $\su_1$, then $\Gamma$ is positive or $\Qcsp(\Gamma)$ is NP-hard.  
\end{lemma}

\begin{proof}
 We first consider the case where $\Gamma$ is preserved by all $\su_i$ with $i \in \N$.
In this case, as we show, $\Gamma$ is a positive language. To this end, we have to show that 
$\{ \su_i \mid i \in \N \}$ generate $\wave$. Let $t$ be an $n$-tuple with pairwise different values
$q_1 < \cdots < q_a < 0 \leq q_{a+1} < \cdots < q_{a+b} \leq 1 < q_{a+b+1} < \cdots < q_{a+b+c}$.
Let $\alpha \in \Aut(\Q;<)$ be such that it sends $q_i$ for $i \in [a]$ to $i-1$, 
$q_i$ for $i \in \{a+1, \ldots, b\}$ to the interval $[a, a+1)$ and $q_i$ for $i \in \{b+1, \ldots, c\}$
to $a+i$ and $\beta$ such that it sends $\su_n(\alpha(q_i))$ for $i \in [a]$ to $q_i$;
$\su_n(\alpha(q_{a+1}))$ to $0$ and $\su_n(\alpha(q_{a+b+i}))$ for $i \in [c]$
to $(q_{a+b+i} - 1)$. Observe that $\beta(\su_n(\alpha(t))) = \wave(t)$.

The second case holds if there exists $k \in \N$ such that $\Gamma$ is preserved by $\su_k$
but is not preserved by $\su_{k+1}$. Let $R \in \Gamma$ and $t \in R$ such that 
$t_s = \su_{k+1}(t)$ is not in $R$. For the sake of simplicity assume that $t$ is an injective tuple,
that is, all its entries are pairwise different. Observe that $t$ has to be of length $n > (k+1)$.
Let $\Pi_1, \ldots, \Pi_{k+1}$ be a partition of $[n]$ such that $i,j \in \Pi_a$ for $a \in [l]$
if and only if $t_s[i] = t_s[j]$. Consider the relation $R_s$ pp-defined by 
$R(x_1, \ldots, x_n) \wedge \bigwedge_{a \in [k+1]} \bigwedge_{i,j \in \Pi_a} x_i = x_j$.
Intuitively, $R_s$ is just $R$ where coordinates from the same $\Pi_i$ are identified.
 Since $R$ and also $R_s$ are preserved by $\su_k$
and $R_s$ is of arity $k+1$, it is easy to see that $R_s$ is preserved by $\wave$ and hence it is a positive
relation. By Theorem~\ref{thm:positivefrontier}, it follows that a positive relation $R$ gives rise
to NP-hard QCSP unless it is preserved by $\pp$ or $\dpp$. 
We now show that $R_s$ is preserved by none of these operations, which imply that $\Qcsp(\Q;R_s)$ and hence $\Qcsp(\Gamma)$ is 
NP-hard and completes the proof of the lemma. Since $t$ is in $R$ and $R$ is preserved by $\su_k$, it follows that $R$
contains also both 
$t_1$ such that $t_1[\Pi_1] < t_1[\Pi_2 \cup \Pi_3] < \cdots < t[\Pi_{k+1}]$
and $t_2$ such that $t_2[\Pi_1 \cup \Pi_2] < t[\Pi_3] < \cdots < t[\Pi_{k+1}]$.
[NOTATION: We write $t[S]$ where $S \subseteq [n]$ for the value $t[i]$ that is common for all $i \in S$.]
By the pp-definition,
we have that $R_s$
contains tuples $t_1^s$ and $t_2^s$ such that $t_1^s[1] < t_1^s[2] = t_1^s[3] < \cdots < t_1^s[k+1]$
and $t_2^s[1] = t_2^s[2] < t_2^s[3] < \cdots < t_2^s[k+1]$. Let now $\alpha \in \Aut(\Q; <)$
be such that $\alpha(t_1^s[1]) < 0 < \alpha(t_1^s[2])$.
Observe now that $t_3^s = \pp(\alpha(t_1^s), t_2^s))$ satisfies 
$t_3^s[1] < t_3^s[2] < t_3^s[3] < \cdots < t_3^s[k+1]$. Hence there is $\beta \in \Aut(\Q; <)$ such that
 $\beta(t_3^s) = t^s$. It contradicts the fact that $t^s$ is not in $R$. It follows that $R_s$
is not preserved by $\pp$. To show that the relation is not preserved by $\dpp$,
we proceed in the similar way with the difference that we take $\alpha$
such that $\alpha(t_2^s[1]) < 0 < \alpha(t_2^s[2])$ and $t_3^s = \pp(\alpha(t_2^s), t_1^s))$.
\end{proof}

\noindent
The next case  to consider is where $\Gamma$ is preserved by $\peak$.

\begin{lemma}
\label{lem:peakallperm}
Let $\Gamma$ be a temporal language preserved by $\peak$. If $\Gamma$ defines neither $\EqXor$ nor 
$\EqOr{n}$ for any $n \geq 3$, then $\Gamma$ is preserved by all permutations. 
\end{lemma}

\begin{proof}
We need some definitions. Let $n$ be a natural number, $t$ 
an $n$-ary tuple and $S \subseteq [n]$. If for all $i, j \in S$,
we have $t[i] = t[j]$, then we write $t[S]$ to indicate the value which is common for 
all $t[i]$ with $i \in S$.
We say that an $n$-ary tuple $t$ is
an \emph{ordered $k$-partition} of $[n]$ if there is an underlying partition $\{ \Pi_1, \ldots, \Pi_k \}$ of $[n]$
such that for all $i,j \in [n]$ we have $t[i] = t[j]$ if and only if $i,j \in \Pi_l$ for some $l \in [k]$
and $t[\Pi_i] < t[\Pi_{i+1}]$ for all $i \in [k-1]$. Let $t_1, t_2$ be two ordered $k$-partitions of $[n]$
with the same underlying partition $\{ \Pi_1, \ldots, \Pi_k \}$ of $[n]$. Then there is an automorphism $\alpha \in \Aut(\Q;<)$
such that $t_1 = \alpha(t_2)$. Hence for an $n$-ary relation $R$, we have $t_1 \in R$ if and only if $t_2 \in R$.
We will write $t[\Pi_1] < \cdots < t[\Pi_k]$ for an ordered  $k$-partition of $[n]$ with an underlying
partition $\{ \Pi_1, \ldots, \Pi_k \}$ of $[n]$. 
We also a need a special treatment of ordered $2$-partitions. We say that a tuple $t$ is an $[a,b]$ $2$-partition if $\left| \Pi_1 \right| = a$
and $\left| \Pi_2 \right| = b$.

If $\Gamma$ pp-defines neither 
$\EqXor$ nor $\EqOr{n}$ for any $n \geq 3$, 
then there are operations $f_{x}$ and $f_1, f_2, \ldots$ preserving $R$
 such that $f_x$ violates $\EqXor$ while $f_n$ for $n \geq 3$ violates $\EqOr{n}$.
Let $R$ be an $n$-ary relation and $t$ a tuple in $R$. For the sake of simplicity
we assume that the values in $t$ are pairwise different. To show that all permutation of $t_o$ are in $R$ we will
prove that for all $k \leq n$, the relation $R$ contains all ordered $k$-partitions of $n$.

\paragraph{First Part of the Proof.}
In the first part of the proof, we prove that $R$ contains all $2$-partitions of $n$.
  By induction on $m = \min(a, b)-1$ we show that 
every $[a,b]$ $2$-partition of $[n]$ is in $R$. If $m = 0$, then $a = 1$ or  $b = 1$.
In this case we just use the operation $\peak$ and automorphisms of $\Aut(\Q; <)$.
Indeed, let $t[[n]\setminus \{i\} ] < t[\{ i \}]$ be an ordered $[n-1,1]$ $2$-partition of $[n]$.
Observe that $t$ is equal to $\beta(\peak(\alpha(s)))$ where $\alpha, \beta$ are automorphisms of $(\Q; <)$ such that $\alpha$
sends $s[i]$ to $0$; and $\beta$ sends $-1$ and $1$ to $t[[n]\setminus \{i\} ]$ and  $t[\{ i \}]$, respectively.   
Now from $t[[n]\setminus \{i\} ] < t[\{ i \}]$ we can obtain any $[1, n-1]$ $2$-partition of $[n]$
by first sending $t[[n]\setminus \{i\} ]$ and $t[\{ i \}]$ to $0$ and $1$, respectively, it flips the values; and then by using an appropriate
automorphism of $(\Q; <)$. Assume now that we are done for $[a,b]$ $2$-partitions with $m$, as defined above, equal to $l$.
We will now prove that the claim holds for $l+1$. 
By the observation above, we have that $R$ is preserved by $f_x$.
Since $f_x$ violates $\EqXor$, by Theorem~\ref{thm:GaloisTemp} and  Lemma~\ref{lem:arityred}, there are tuples $s, s_1, s_2, s_3, s_4, s_5$ 
such that 
\begin{itemize}
\item $s_1[1] = s_1[2] = s_1[3]$,
\item $s_2[1] = s_2[2] < s_2[3]$,
\item $s_3[1] = s_3[2] > s_3[3]$,
\item $s_4[1] = s_4[3] < s_4[2]$,
\item $s_5[1] = s_5[3] > s_5[2]$.
\end{itemize}
and we have $f(s_1, s_2, s_3, s_4, s_5) = s$ and $s$ such that $s[1] \neq s[2]$ and $s[1] \neq s[3]$.

Let $t_g[\Pi_1] < t_g[\Pi_2]$ be any ordered $[n-l-1, l+1]$ ordered $2$-partition of $[n]$.
Let $i \in \Pi_i$.
Since $R$ is preserved by a constant operation we have that it contains a tuple 
$t_1$ such that all its entries are equal to $s_1[1]$.
Moreover, by the induction assumption, we have that $R$ contains all of the following:
\begin{itemize}
 \item an ordered $[n-l, l]$ $2$-partition 
$t_2[\Pi_1 \cup \{ i \}  ]  = s_1[1] < t_2[\Pi_2 \setminus \{ i \}] = s_1[3]$;
\item an ordered $[l, n-l]$ $2$-partition
$t_3[\Pi_2 \setminus \{ i \}] = s_2[3] < t_3[\Pi_1 \cup \{ i \}  ]  = s_2[1]$;
\item an ordered $[n-1, 1]$ $2$-partition
$t_4[[n] \setminus \{i\}] = s_4[1] < t_4[i] = s_4[2]$; and
\item an ordered $[1, n-1]$ $2$-partition
$t_5[i] = s_5[2] < t_5[[n] \setminus \{i\}] = s_5[1]$.
\end{itemize}

It is now straigthforward to check that $t = f(t_1, t_2, t_3, t_4, t_5)$
satisfies $t[\Pi_1] \neq t[\Pi_2 \setminus \{i\}]$ and 
$t[\Pi_1] \neq t[\{i\}]$. It is now easy to see
that by applying  $\peak$ and appropriate automorphisms of $(\Q; <)$
we can obtain $t_g[\Pi_1] < t_g[\Pi_2]$. This proves that we can obtain
any $[n-l-1, l+1]$ ordered $2$-partition of $[n]$. Any 
$[l+1, n-l-1]$ ordered $2$ partition of $[n]$ can be obtained from
a corresponding $[n-l-1, l+1]$ ordered partition just by flipping the values.
This can be obtained with the use of $\peak$. This proves that $R$ 
contains all ordered $[a,b]$ $2$-partitions where $\min(a,b) = l+1$
and completes the induction step. By mathematical induction we obtain that $R$ 
contains all ordered $2$-partitions.

\paragraph{Second Part of The Proof.}
Here we show that $R$ contains all ordered $k$-partitions of $[n]$
for $k \leq n$. By the previous part of the proof we have that $R$ contains all 
$2$-partition of $n$, which we will use as a base case in our induction.
Assume now that we are done for $k = l$. We will now prove the claim for $k = l+1$.
Recall that $f_{l+1}$ is an operation preserving $\Gamma$ and violating 
$\EqOr{l+1}$. Let $s_1, \ldots, s_p$ be a list of all, up to isomorphisms, 
different, ordered $a$-partition of $[l+1]$ with $a \leq l$. Observe that up to a permutation of $s_1, \ldots, s_p$,
we can assume that $f(s_1, \ldots, s_p) = s$ 
and $s$ satisfies $s[1] < \cdots < s[l+1]$. 
Let now $t_g[\Pi_1] < \cdots < t_g[\Pi_{l+1}]$ be any ordered 
$(l+1)$-partition of $[n]$. We will now show that $t_g$ is in $R$.
Let $t_1, \ldots, t_p$ be a list of  all, up to isomorphisms,  
ordered $a$-partitions of $[n]$ with $a \leq l$ such that
for all $b \in [p]$ and $i \in [l+1]$ we have $t_b[i] = s_b[\Pi_i]$.
By the induction hypothesis we have that all tuples 
$t_1, \ldots, t_p$ are in $R$.
Observe that there is an automorphism $\alpha \in \Aut(\Q;<)$
such that $\alpha(f(t_1, \ldots, t_p)) = t_g$. This completes the induction step. By mathematical induction we have that
every ordered $k$-partition with $k \leq n$ is in $R$. This completes the proof of the lemma.
\end{proof}

The only remaining polymorphisms from  Lemma~\ref{lem:violBetwCconst}
are $-$ and $\cyc$. We take care of them in three steps. First we look at the case where $\Gamma$
has only $-$ out of these two.

\begin{lemma}
\label{lem:violBetwC}
Let $\Gamma$ be a temporal language preserved by a constant operation and by $-$.
If $\Gamma$ does not pp-define $\BetwC$, then it is preserved by $\su, \peak, \ci, \ic, \cyc$ or all permutations.
\end{lemma}

\begin{proof}
If $\Gamma$ does not pp-define $\BetwC$, then by Theorem~\ref{thm:GaloisTemp},
there is an operation $f_b$ that preserves $\Gamma$ and violates $\BetwC$. By Lemma~\ref{lem:arityred},
we can assume that $f_b$ is a ternary operation such that  for some tuples $t_1, t_2, t_3$ satisfying 
$t_1[1] < t_1[2] < t_1[3]$, $t_2[3] < t_2[2] < t_2[1]$,
and $t_4[1] = t_4[2] = t_4[3]$
we have $f_b(t_1,t_2,t_3) = t$ and $t \notin \BetwC$.
Since $\Gamma$ is preserved by a constant operation, it is preserved by all constant operations, in particular
some constant operation $f_c$ satisfying $f_c(t_1, t_2, t_3) =  t_4$. By Observation~\ref{obs:futherred},
we have that $\Gamma$ is preserved by an operation $f_d: \Q^3 \rightarrow \Q$ such that 
$f_d(t_1, t_2) = t$.
There are also  an automorphism $\alpha \in \Aut(\Q;<)$ such that $\cyc(\alpha(t_1)) = t_2$
and $\cyc(\alpha(t_1)) = t_3$. Thus, by applying Observation~\ref{obs:futherred} again, 
we conclude that $\Gamma$
is preserved by a unary operation $f$ such that $f(t_1) = t$. Observe that the operation $f$ violates $\Betw$
and $<$ and is not a constant operation. We us Lemma~\ref{lem:noncgen}. First we consider the case where
 $f$ is injective. Since it violates both $\Betw$ and $<$, we have by
Item~\ref{noncgen:inj} that $f$ generates $\cyc$ or all permutations. 
On the other hand, if $f$ is not injective, then it generates $\ic, \ci, \su_1$, or $\peak$.
It completes the proof of the lemma.
\end{proof}

\noindent
We now take care of the situation where $\Gamma$ is preserved by $\cyc$.
%We are in one of the previously considered cases or $\Gamma$
%is preserved by both $\cyc$ and $-$.

\begin{lemma}
\label{lem:violCyclC}
Let $\Gamma$ be a temporal language preserved by a constant operation and by $\cyc$.
If $\Gamma$ does not pp-define $\CyclC$, then it is preserved by $\su_1, \peak, \ci, \ic, -$ or all permutations.
\end{lemma}

\begin{proof}
If $\Gamma$ does not pp-define $\CyclC$, then by Theorem~\ref{thm:GaloisTemp}, it follows that
there is an operation $f_b$ that preserves $\Gamma$ and violates $\CyclC$. By Lemma~\ref{lem:arityred},
we can assume that $f_b$ is an operation of arity four such that some for tuples $t_1, t_2, t_3, t_4$ satisfying 
$t_1[1] < t_1[2] < t_1[3]$, $t_2[2] < t_2[3] < t_2[1]$, $t_3[3] < t_3[1] < t_3[2]$,
and $t_4[1] = t_4[2] = t_4[3]$
we have $f_b(t_1,t_2,t_3, t_4) = t$ and $t \notin \BetwC$.
Since $\Gamma$ is preserved by a constant operation, it is preserved by all constant operations, in particular
some constant operation $f_c$ satisfying $f_c(t_1, t_2, t_3) =  t_4$. By Observation~\ref{obs:futherred},
we have that $\Gamma$ is preserved by an operation $f_d: \Q^3 \rightarrow \Q$ such that 
$f_d(t_1, t_2, t_3) = t$.
There are also  automorphisms $\alpha, \beta \in \Aut(\Q;<)$ such that $\cyc(\alpha(t_1)) = t_2$
and $\cyc(\alpha(t_1)) = t_3$. Thus, by applying Observation~\ref{obs:futherred} twice, 
we conclude that $\Gamma$
is preserved by a unary operation $f$ such that $f(t_1) = t$. Observe that $f$
violates $<$ and $\Cycl$, defined in the formulation of Lemma~\ref{lem:noncgen}.
If $f$ is injective, then by Item~\ref{noncgen:inj} of Lemma~\ref{lem:noncgen},
we have that $f$ generates $-$ or all permutations. 
The operation $f$ is also not constant. Thus, if it is not injective, then 
by the same lemma, it follows that it generates $\ic,\ci, \su_1$ or $\peak$.
This completes the proof of the lemma.
\end{proof}

\noindent
Finally, we consider the situation where $\Gamma$ is preserved by both $-$ and $\cyc$.

\begin{lemma}
\label{lem:violI}
Let $\Gamma$ be a temporal language preserved by a constant operation by $-$, and $\cyc$.
If $\Gamma$ does not pp-define $S$, then it is preserved by $\su, \peak, \ci, \ic$ or all permutations.
\end{lemma}

\begin{proof}
The proof goes along the lines of the proof of Lemmas~\ref{lem:violBetwC} and~\ref{lem:violCyclC}. 
Again, if $\Gamma$
does not pp-define $S$, then there is  an operation $f_b$ that preserves $\Gamma$ and violates $S$.
Again, we can assume that the arity of $f_b$ is the number of orbits of $3$-tuples with respect to $\Aut(\Q;<)$ 
contained in $S$. Thus there are tuples $t_1, t_2, t_3, t_4, t_5, t_6, t_7$
such that $t_1[1] < t_1[2] < t_1[3]$, $t_2[1] < t_2[3] < t_1[2]$, $t_3[2] < t_3[1] < t_3[3]$,
$t_4[2] < t_4[3] < t_4[1]$, $t_5[3] < t_5[1] < t_5[2]$, $t_6[3] < t_6[2] < t_6[1]$, and
$t_7[1] = t_7[2] = t_7[3]$. Observe that tuples $t_2, \ldots, t_7$ can be obtained from $t_1$
by applying to $t_1$ the operations: $-. \cyc$, constant operations and automorphisms of $(\Q;<)$.
Thus by multiple application of Observation~\ref{obs:futherred}, we conclude that there is a unary operation
$f$ that preserves $\Gamma$ and such that $f(t_1) = t$ for some $t \notin S$.
Observe that $f$ violates $<, \Betw, \Cycl$ and is not a constant operation.
If it is injective, then by Item~\ref{noncgen:inj} of Lemma~\ref{lem:noncgen},
it follows that $f$ generates all permutations. If $f$ is not injective, then by the same lemma, we have that
$f$ generates $\ic, \ci, \su_1$ or $\peak$.
\end{proof}

\noindent
We can use the above lemmas to prove the following. 

\begin{theorem}
\label{thm:constthm}
Let $\Gamma$ be a dually-closed temporal constraint language preserved by a constant operation, 
then one of the following holds.
\begin{enumerate}
\item \label{constthm:hard} The problem $\Qcsp(\Gamma)$ is coNP-hard or NP-hard.
\item \label{constthm:sdOH} $\Gamma$ is a dually-closed Ord-Horn constraint language.
\item \label{constthm:pos} $\Gamma$ is a positive constraint language.
\item \label{constthm:eq} $\Gamma$ is an equality constraint language.
\end{enumerate}
\end{theorem}

\begin{proof}
By Lemma~\ref{lem:violBetwCconst}, we have that either $\Gamma$ pp-defines $\BetwC$ and then by 
Theorem~\ref{thm:somehard}, the problem $\Qcsp(\Gamma)$ is coNP-hard;
or $\Gamma$ is preserved by one of the following operations: $\pp, \dpp, \lele, \dlele,
\su_1, \ic, \ci, \peak, -, \cyc$, or all permutations.  If it is one of the first four operations, then by 
Lemma~\ref{lem:llppOH} we have that $\Gamma$ is a dually-closed Ord-Horn language and we are in case~\ref{constthm:sdOH}.

Observe now that a dually-closed language preserved by $\ic$ or $\ci$ is preserved by both $\ic$ and $\ci$. 
These operations are the dual of each other.
Moreover,  $\{\ic, \ci\}$ generate $\su_1$. Indeed, for every tuple $t$ with pairwise different values
$q_1 < \cdots < q_a \leq 0 < q_{a+1} < \cdots < q_{a+b}$ we have that $\ci(\alpha(\ic(t))) = \su_1(t)$
where $\alpha$ is an automorphism of $(\Q;<)$ such that $\alpha(q_1) < \cdots < \alpha(q_a) < 0$ and $\alpha(O) = 1$.
It follows by Lemma~\ref{lem:su}, that if a dually-closed temporal constraint language $\Gamma$ 
preserved by $\su_1, \ic$, or $\ci$, then $\Gamma$ is either hard and we are in Case~\ref{constthm:hard} or
positive and we are in Case~\ref{constthm:pos}. Further, by Lemma~\ref{lem:peakallperm} and Theorem~\ref{thm:somehard}, 
if $\Gamma$ is preserved by $\peak$, then $\Gamma$ is either hard and we are in Case~\ref{constthm:hard} or
$\Gamma$ is an equality language and we are in case~\ref{constthm:eq}. What remained to consider is the situation 
where $\Gamma$
is preserved by $-$ or $\cyc$. 
In the former case we use Lemma~\ref{lem:violBetwC} and
Lemma~\ref{lem:violCyclC} in the latter. We have that either $\Gamma$ pp-defines a relation that
give rise to hard $\Qcsp$ or it is preserved by both $-$ and $\cyc$. In the first case we are in
Case~\ref{constthm:hard} and we are done, whereas in the second we use Lemma~\ref{lem:violI}.
Here, again, either $\Gamma$ pp-defines $S$ and by Theorem~\ref{thm:somehard}, we are in 
Case~\ref{constthm:hard}; or one of the previously considered cases holds and we are also done.
 This completes the proof of the theorem. 
\end{proof}

\section{Classification}
\label{sect:classs}

Here, we prove that $\Qcsp(\Gamma)$ for a dually-closed temporal language $\Gamma$ is
either hard or it is in $P$.

\begin{theorem}
\label{thm:selfdualclass}
Let $\Gamma$ be a dually-closed temporal language.
Then $\Gamma$ is a Guarded Ord-Horn language and $\Qcsp(\Gamma)$ is in P.
Otherwise 
$\Qcsp(\Gamma)$ is NP-hard or coNP-hard.
\end{theorem}

\begin{proof}
By Theorem 50 in~\cite{tcsps-journal}, it follows that $\Csp(\Gamma)$, and hence also $\Qcsp(\Gamma)$ is 
hard,  or $\Gamma$ is preserved by: $\pp, \dpp, \lele, \dlele$ or a constant operation.
In first four cases we use Lemma~\ref{lem:llppOH} which reduces the problem to the classification 
of dually-closed Ord-Horn temporal constraint satisfaction problems. In the case where $\Gamma$ 
is preserved by a constant operation, by Theorem~\ref{thm:constthm},
$\Qcsp(\Gamma)$ is hard or $\Gamma$ is either a dually-closed Ord-Horn language or
an equality language, or a positive language. In the first two cases, it is in fact dually-closed Ord-Horn
and hence by Theorem~\ref{thm:duallyOH},
the language $\Gamma$ is Guarded Ord Horn and $\Qcsp(\Gamma)$ is in $P$ or $\Qcsp(\Gamma)$ is coNP-hard.
We now consider the case where $\Gamma$ is positive. By Theorem~\ref{thm:positivefrontier},
we have that either $\Qcsp(\Gamma)$ is NP-hard, or $\Gamma$ is preserved by $\pp$ or $\dpp$.
In the former case we are done. In the latter case, by Lemma~\ref{lem:llppOH} 
we have that $\Gamma$ is dually-closed Ord-Horn and we ared one by Theorem~\ref{thm:duallyOH}.
\end{proof}

%\subsubsection{Acknowledgements.}
%The author would like to thank Hubie Chen and Manuel Bodirsky
%for many inspiring discussion on the subject of this paper, and
%the referees for their helpful comments.

\bibliographystyle{alpha}
\bibliography{mybib}

\end{document}